\RequirePackage{amsmath}
\documentclass{llncs} 

\usepackage{amssymb}
\usepackage{textcomp}
\usepackage{comment}
\usepackage{color}
\newcommand{\Mod}[1]{\ \mathrm{mod}\ #1}

\newcommand{\ket}[1]{|#1\rangle}
\newcommand{\braket}[2]{\langle #1|#2\rangle}

\newcommand{\dollar}[0]{\$}

\newcommand{\ModXOR}{\mathtt{MODXOR_k}}
\newcommand{\SSA}{\mathtt{SSA_n}}

\newtheorem{fact}{Fact}

\begin{document}

\title{Zero-Error Affine, Unitary, and Probabilistic OBDDs\thanks{Part of the research work was done while Ibrahimov was visiting University of Latvia in February 2017.}}
\author{Rishat Ibrahimov$^{1}$   \and Kamil~Khadiev$^{1,2}$ \and Kri\v{s}j\={a}nis Pr\={u}sis$^{2}$ \and Jevgenijs Vihrovs$^{2}$ \and Abuzer Yakary{\i}lmaz$^{2}$ } 

\institute{Kazan Federal University, Kazan, Russia  \and University of Latvia, R\={\i}ga, Latvia \\ \email{rishat.ibrahimov@yandex.ru, kamilhadi@gmail.com, krisjanis.prusis@lu.lv, jevgenijs.vihrovs@lu.lv,abuzer@lu.lv} 
}

\maketitle

\begin{abstract}
	We introduce the affine OBDD model and show that zero-error affine OBDDs can be exponentially narrower than bounded-error unitary and probabilistic OBDDs on certain problems. Moreover, we show that Las Vegas unitary and probabilistic OBDDs can be quadratically narrower than deterministic OBDDs. We also obtain the same results by considering the automata versions of these models.
    
\end{abstract}
\section{Introduction}

Affine finite automata (AfA) were introduced as a quantum-like ``almost'' linear\footnote{It evolves linearly but a non-linear operator is applied when we retrieve information from the state vector.} classical system that has been examined in  a series of papers, in which they were compared with classical and quantum finite automata \cite{dy2016,vy2016,BMY16A,hmy2017}. Both bounded- and unbounded-error AfAs have been shown to be more powerful than their probabilistic and quantum counterparts and they are equivalent to quantum models in nondeterministic acceptance mode \cite{dy2016}. AfAs can also be very succinct on languages and promise problems \cite{vy2016,BMY16A}. Here we present the power of zero-error affine computation on total functions and languages.
 
Ordered Binary Decision Diagrams (OBDD), also known as oblivious read once branching programs \cite{Weg00}, are an important restriction of branching programs.  Since the length of an OBDD is fixed (linear), the main complexity measure is ``width'', analogous to the number of states used by automata models. OBDDs can also be seen as a variant of nonuniform automata that allow accessing the input in a predetermined order and using possibly different sets of instructions in each step (see \cite{ag05}).  
It is known that \cite{agkmp2005,agky14,agky16,ak96,a97,g15,ss2005,kk2017} the gap between the width of the following OBDD models can be at most exponential on total functions: deterministic and bounded-error probabilistic, deterministic and bounded-error unitary, and bounded-error probabilistic and bounded-error unitary. Each exponential gap has also been shown to be tight by giving a function which achieves it. On partial functions, on the other hand, the width gap between deterministic and exact unitary OBDDs can be unbounded \cite{agky14,g15,agky16}. 

In this paper, we introduce affine OBDDS and then compare error-free (zero-error or Las Vegas) affine, unitary, probabilistic, and deterministic OBDD models. Zero-error probabilistic OBDDs are identical to deterministic OBDDs. Zero-error unitary OBDDs cannot be narrower than deterministic OBDDs. Surprisingly, zero-error affine OBDDs can be exponentially narrower than not only deterministic OBDDs but also bounded-error probabilistic and unitary OBDDs.

With success probability $ \frac{1}{2} $, Las Vegas probabilistic and unitary OBDDs can be quadratically narrower than deterministic OBDDs. We give an example function that achieves this bound up to a logarithmic factor. 

Finally, we examine the automata counterpart of these models and obtain almost the same results. 

Preliminaries and the definitions of the used models are given in the next section. The generic lower bounds are given in Section \ref{sec:lower-bounds}. The results regarding zero-error affine OBDD are given in Section \ref{sec:exact-obdd} and the results regarding Las Vegas OBDDs are given in Section \ref{sec:lv-obdd}. We close the paper with automata results in Section \ref{sec:lv-FA}. 

\section{Preliminaries}\label{sec:prelims}

We assume the reader familiar with the basics of branching programs, refer to Appendix \ref{app:basics-branching} for the details.

An oblivious leveled read-once branching program is also called an Ordered Binary Decision Diagram (OBDD).
An OBDD $P$ reads the variables in a specific order $\pi=(j_1,\dots,j_n)$. 
We call $\pi$ the order of $P$. An OBDD can also be seen as a non-uniform counterpart of a finite automaton that can read the input symbols in a predetermined order. Therefore an OBDD can trace its computation on a finite set of states $ S = \{ s_1,\ldots,s_m \} $ such that (i) $ m $ is the width of OBDD, (ii) the initial state is the node $ s $, (iii) the accepting states are the accepting sink nodes. Thus, each node in any level can be easily associated with a state in $ S $. Besides, an OBDD can have a different transition function at each level. 
 
A probabilistic OBDD (POBDD) $P_n$ of width $ m $ is a 5-tuple:

$
	P_n = (S,T,v_0,S_a,\pi),
$  
where $ S = \{s_1,\ldots,s_m\} $ is the set of states corresponding to at most one node in each level, $ v_0 $ is the initial probabilistic state that is a stochastic column vector of size $ m $, $ S_a $ is the set of accepting states corresponding to the accepting sink nodes in the last level, $ \pi $ is a permutation of $ \{1,\ldots,n\} $ defining the order of the input variables, and, $ T = \{T^0_i,T^1_i \mid 1 \leq i \leq n \} $ is the set of (left) stochastic transition matrices such that at the $ i $-th step $ T_i^0 $ ($T_i^1$) is applied if the corresponding input bit is 0 (1). 

Let $ x \in \{0,1\}^n $ be the given input. At the beginning of the computation $ P_n $ is in $ v_0 $. Then the input bits are read in the order $ \pi(1), \pi(2),\ldots,\pi(n) $ and the corresponding stochastic operators are applied:
$
	v_j = T_j^{x_{\pi(j)}} v_{j-1}.
$
This represents the transition at the $ j $-th step, where $ v_{j-1} $ and $ v_j $ are the probabilistic states before and after the transition respectively, $ \pi(j) $ represents the input bit read in this step, $ x_{\pi(j)} $ is the value of this bit, and $ T_j^{x_{\pi(j)}} $ is the stochastic operator applied in this step. We represent the final state as $ v_f = v_n $. The input $ x $ is accepted (the value of 1 is given) with probability
\[
	f_{P_n}(x) = \sum_{s_i \in S_a} v_f[i].
\]

If all stochastic elements of a POBDD are composed of only 0s and 1s, then it is a deterministic OBDD.

Quantum OBDDs (QOBDD) using superoperators are non-trivial generalizations of POBDDs \cite{agky16}. In this paper, we use the most restricted version of quantum OBDDs called unitary OBDDs (UOBDDs) \cite{GaiY16B}. Remark that UOBDDs and POBDDs are incomparable \cite{ss2005}. (We refer the reader to \cite{SayY14} for a pedagogical introduction to the basics of quantum computation.)

Formally, a UOBDD, say $ M_n $, with width $ m $ is a 5-tuple

$
	M_n = (Q,T,\ket{v_0},Q_a,\pi),
$
where $ Q = \{ q_1,\ldots,q_m \} $ is the set of states, $ \ket{v_0} $ is the initial quantum state, $Q_a$ is the set of accepting states, $ \pi $ is a permutation of $ \{1,\ldots,n\} $ defining the order of the variables, and $ T = \{ T_i^0,T_i^1 \mid 1 \leq i \leq n \} $ is the set of unitary transition function matrices such that at the $i$-th step $ T_i^0 $ ($ T_i^1 $) is applied if the corresponding input bit is 0 (1).

Let $ x \in \{0,1\}^n $ be the given input. At the beginning of the computation $ M_n $ is in $ \ket{v_0} $. Then the input bits are read in the order $ \pi(1), \pi(2),\ldots,\pi(n) $ and the corresponding unitary operators are applied:
$
	\ket{v_j} = T_j^{x_{\pi(j)}} \ket{v_{j-1}}.
$
This represents the transition at the $ j $-th step, where $ \ket{v_{j-1}} $ and $ \ket{v_j} $ are the quantum states before and after the transition respectively, $ \pi(j) $ represents the input bit read in this step, $ x_{\pi(j)} $ is the value of this bit, and $ T_j^{x_{\pi(j)}} $ is the unitary operator applied in this step. We represent the final state as $ \ket{v_f} = \ket{v_n} $. At the end of the computation, the final state is measured in the computational basis and the input is accepted if the observed state is an accepting one. Thus, the input $ x $ is accepted with probability
$
	f_{M_n}(x) = \sum_{q_i \in Q_a} |\braket{q_i}{v_f}|^2,
$
where $ \ket{q_i} $ represents the basis state corresponding to state $ q_i $ and $ \braket{q_i}{v_f} $ returns the value of $ q_i $ in the final state. 

An $ m $-state affine system \cite{dy2016} can be represented by the space $ \mathbb{R}^m $. The set of (classical) states is denoted $ E = \{e_1,\ldots,e_m\} $. Any affine state (similar to a probabilistic state) is represented as a column vector
\[
	v = \left( \begin{array}{c} \alpha_1 \\ \alpha_2 \\ \vdots \\ \alpha_m 
\end{array}	 \right) \in \mathbb{R}^m	
\]  
such that $ \sum_{i=1}^m \alpha_i = 1 $. Each $e_i$ also corresponds to a standard basis vector of $ \mathbb{R}^m $ having value $1$ in its $i$-th entry. An affine operator is an $ m \times m $ matrix, each column of which is an affine state, where the $ (j,i) $-th entry represents the transition from state $ e_i $ to state $ e_j $. If we apply an affine operator $ A $ to the affine state $ v $, we obtain the new affine state $ v' = Av $. To get information from the affine state, a non-linear operator called weighting is applied, which returns each state $e_i$ with probability equal to the weight of the corresponding vector element in the $l_1$ norm of the affine state. 
If it is applied to $ v $, the state $ e_i $ is observed with probability
$
	\frac{|\alpha_i|}{\sum_{j=1}^n |\alpha_j|}
	= 
	\frac{|\alpha_i|}{|v|},
$
where $ |v| $ is the $ l_1 $ norm of $v$.

Here we define affine OBDDs (AfOBDDs) as a model with both classical and affine states, which is similar to the quantum model having quantum and classical states \cite{AW02}. This addition does not change the computational power of the model, but helps in algorithm construction.

Formally, an AfOBDD, say $ M_n $, having $ m_1 $ classical and $ m_2 $ affine states is an 9-tuple

$
	M = (S,E,\delta,T,s_I,v_0,S_a, E_a,\pi),
$
where $ S = \{s_1,\ldots,s_{m_1}\} $ is the set of classical states, $ s_I \in S $ is the initial classical state, $ S_a \subseteq S $ is the set of classical accepting states, $ E = \{e_1,\ldots,e_{m_2}\} $ is the set of affine states, $ v_0 \in \mathbb{R}^{m_2} $ is the initial affine state, $E_a \subseteq E$ is the set of affine accepting states, $ \pi $ is a permutation of $ \{1,\ldots,n\} $ defining the order of the variables, $ \delta = \{\delta_i \mid 1 \leq i \leq n\} $ is the classical transition function such that at the $i$-th step the classical state is set to $ \delta(s,x_{\pi(i)}) $ when in state $s \in S$ and corresponding input bit is $ x_{\pi(i)} $, and, $ T = \{T^{s_j,0}_i,T^{s_j,1}_i \mid s_j \in S \mbox{ and } 1 \leq i \leq n \} $ is the set of affine transition matrices such that at the $ i $-th step $ T_i^{s_j,0} $ ($T_i^{s_j,1}$) is applied if the corresponding input bit is 0 (1) and the current classical state is $ s_j $. The width of $ M_n $ is equal to $ m_1 \cdot m_2 $. 

Let $ x \in \{0,1\}^n $ be the given input. At the beginning of the computation $ M_n $ is $(s_I, v_0) $. Then the input bits are read in the order $ \pi(1), \pi(2),\ldots,\pi(n) $. In each step, depending on the current input bit and classical state, the affine state is updated and then the classical state is updated based on the current input bit.  Let $ (s,v_{j-1}) $ be the classical-affine state pair at the beginning of the $ j $-th step. Then the new affine state is updated as
$
	v_j = T_j^{s,x_{\pi(j)}} v_{j-1}.
$
After that the new classical state is updated by $ \delta_j(s,x_{\pi(j)}) $. At the end of the computation we have $ (s_F,v_f) $. If $ s_F \notin S_a $, then the input is rejected. Otherwise, the weighting operator is applied to $v_f$, and the input is accepted with probability
\[
	f_{M_n}(x) = \sum_{e_i \in E_a} \frac{|v_f[i]|}{|v_f|}.
\]

Remark that if an AfOBDD restricted to use non-negative numbers, then we obtain a POBDD.

Any OBDD with $ \pi = (1,\ldots,n) $ is called an id-OBDD. If we use the same transitions at each level of an id-OBDD, then we obtain a finite automaton (FA). A FA can also read an additional symbol after reading the whole input called the right end-marker ($\dollar$) for the post-processing.
We abbreviate the FA versions of OBDD, POBDD, QOBDD, UOBDD, and AfOBDD as DFA, PFA, QFA, UFA, and AfA, respectively. Remark that UFAs are also known as Measure-Once or Moore-Crutchfield quantum finite automata \cite{MC00,AY15}.

A Las Vegas automaton is one that can make three decisions -- ``accept'', ``reject'' and ``don't know''. Therefore, its set of states is split into three disjoint sets, the set of accepting, rejecting and neutral states in which the aforementioned decisions are given, respectively. To be a well-defined Las Vegas algorithm, each language member (non-member) must be rejected (accepted) with zero probability.

We assume the reader is familiar with the basic terminology of computing functions and recognizing languages. Here we revise the ones we use in the paper.

A function  $ f:\{0,1\}^n \rightarrow \{1,0\} $ is computed by a bounded-error machine if each member of $ f^{-1}(1) $ is accepted with probability at least $ 1-\epsilon $ and each member of  $ f^{-1}(0) $ is accepted with probability less than $\epsilon $ for some non-negative $ \epsilon < \frac{1}{2} $. If $ \epsilon = 0 $, then it is called exact (zero-error) computation.

In the case of FAs, languages are considered instead of functions and the term ``language recognition'' is used instead of ``computing a function''.

A Las Vegas FA can recognize a language $ L $ with success probability $ p < 1 $ (with error bound $ 1-p $) if each member (non-member) is accepted with probability at least $ p $ (less than $ 1-p $).

\newcommand{\DFA}[1]{\mathsf{DFA(#1)}}
\newcommand{\LV}[1]{\mathsf{LV_\varepsilon(#1)}}
\newcommand{\LVhalf}[1]{\mathsf{LV_{0.5}(#1)}}
\newcommand{\QLV}[1]{\mathsf{QLV_\varepsilon(#1)}}
\newcommand{\QLVhalf}[1]{\mathsf{QLV_{0.5}(#1)}}
\newcommand{\ULV}[1]{\mathsf{ULV_\varepsilon(#1)}}
\newcommand{\ULVhalf}[1]{\mathsf{ULV_{0.5}(#1)}}

\newcommand{\OBDD}[1]{\mathsf{OBDD(#1)}}
\newcommand{\EPOBDD}[1]{\mathsf{exact\!-\!POBDD(#1)}}
\newcommand{\EUOBDD}[1]{\mathsf{exact\!-\!UOBDD(#1)}}
\newcommand{\ULVOBDD}[1]{\mathsf{ULV\!-\!OBDD_{\varepsilon}(#1)}}
\newcommand{\LVOBDD}[1]{\mathsf{LV\!-\!OBDD_{\varepsilon}(#1)}}
\newcommand{\ULVOBDDhalf}[1]{\mathsf{ULV\!-\!OBDD_{0.5}(#1)}}
\newcommand{\LVOBDDhalf}[1]{\mathsf{LV\!-\!OBDD_{0.5}(#1)}}

For a given language $ L $, $ \DFA{L} $, $ \LV{L} $, and $ \QLV{L} $ denote the minimum sizes of a DFA, an LV-PFA and an LV-QFA recognizing the language $ L $, respectively, where the error bound is $ \varepsilon $ for the probabilistic and quantum models. For a given Boolean function $ f $, $ \OBDD{f} $, $ \LVOBDD{f} $, and $ \ULVOBDD{L} $ denote the minimum sizes of an OBDD, an LV-POBDD and an LV-UOBDDs computing $ f $, respectively, where the error bound is $ \varepsilon $ for the probabilistic and quantum models. Remark that in the case of zero-error we set $ \varepsilon = 0 $.


\section{Lower Bounds}
\label{sec:lower-bounds}

We start with the necessary definitions and notations.  Let $X = \{X_1, \dots, X_n\}$ be the set of variables. Let $\theta=(X_A,X_B)$ be a partition of the set $X$ into two parts $X_A$ and $X_B=X\backslash X_A$. 
Let  $f|_\rho$ be  a subfunction of $f$, where  $\rho$ is a mapping $\rho:X_A \to \{0,1\}^{|X_A|}$.
The function $f|_\rho(X_B)$ is obtained from $f$ by applying $\rho$. The mapping $\rho$ assigns a Boolean value for each variable from $X_A$.  Let $N^\theta(f)$ be the number of different subfunctions with respect to the partition $\theta$.
Let $\Pi(n)$ be the set of all permutations of $\{1,\dots,n\}$.
Let $\theta(\pi,u)=(X_A,X_B)=(\{X_{j_{1}},\dots,X_{j_u}\},\{X_{j_{u+1}},\dots,X_{j_n}\})$, for the permutation $\pi=(j_1,\dots, j_n)\in \Pi(n), 1<u<n$. We denote $\Theta(\pi)=\{\theta(\pi,u): 1<u<n\}$.
Let $ N^\pi(f)=   \max_{\theta\in \Theta(\pi)} N^\theta(f),$ $
N(f)=\min_{\pi\in \Pi(n)}N^\pi(f). $

Based on techniques from communication complexity theory, it has been shown that exact quantum and probabilistic protocols have at least the same complexity as deterministic ones \cite{Kla00,hs03}. The following lower bounds are also known.

\begin{fact}
	\label{fact:lower-bound-automata}
	\cite{dhr97,Kla00,hs2001,hs03}
	For any regular language $ L $ and error bound $ \varepsilon < 1 $, we have the following lower bounds for PFAs and QFAs:

$
	(\DFA{L})^{1 - \varepsilon} \leq \LV{L}
	\mbox{ and }
	(\DFA{L})^{1 - \varepsilon}\leq \ULV{L}.
$
\end{fact}

\begin{fact}\cite{ss2005} \label{thm:lower-bound-OBDD1}
For any Boolean function $f$ over $X=(X_1,\dots,X_n)$ and error bound $ \varepsilon < 1 $:
	$
 (\OBDD{f})^{1 - \varepsilon}\leq \ULVOBDD{f}.
	$
\end{fact}

We can easily extend this result for the probabilistic OBDD model as well.
\begin{theorem}
	\label{thm:lower-bound-OBDD2}
	For any Boolean function $f$ over $X=\{X_1,\dots,X_n\}$ and error bound $ \varepsilon < 1 $:
	$
		(\OBDD{f})^{1 - \varepsilon} \leq \LVOBDD{f} .
	$
\end{theorem}
\begin{proof}
Let $d=\OBDD{f}$. Due to \cite{Weg00}, we have $N(f)=d$.
Assume that there is a Las Vegas OBDD $P$ of width $w$ such that $w< d^{1 - \varepsilon}$. Let $u=\texttt{argmax}_t N^{\theta(\pi(P),t)}(f)$ and $\theta=\theta(\pi(P),u)$.
Then $P$ can be simulated by a Las Vegas probabilistic protocol (see \cite{k16}) with $\log_2{w}<(1-\varepsilon)\log_2 d$ bits. By the definition of the number of subfunctions, we have $N^{\theta}(f)\geq d$. And it is known that the deterministic communication complexity of a function is $\log_2(N^{\theta}(f))=\log_2 d$. 
But we also have a Las Vegas communication protocol which uses $\log_2{w}<(1-\varepsilon)\log_2 d$ communication bits, giving a contradiction with \cite{hs03}.
In  that paper it was shown that, if the best deterministic communication protocol sends $r$ bits, then a Las Vegas protocol with an $\varepsilon$ probability of giving up  cannot send less than $(1-\varepsilon)r$ bits. 
\qed
\end{proof}

It is trivial that these results imply equivalence for exact (zero-error) computation, where $\varepsilon=0$.

\section{Zero-Error Affine OBDDs}
\label{sec:exact-obdd}

We show that exact AfOBDDs can be exponentially narrower than classical and unitary quantum OBDD models. For this purpose we use three different functions.


\newcommand{\hwb}{\mathtt{HWB_n}}

The hidden weighted bit function \cite{Weg00} $ \hwb :\{0,1\}^n\to\{0,1\}$ returns the value of $x_{z}$ on the input $x=(x_1,\dots,x_n)$ where $z=x_1+\cdots+x_n$, taking $x_0=0$.
It is known \cite{Weg00} that any OBDD solving $ \hwb $ has a minimum width of $ 2^{n/5}/n $. Remark that due to Fact \ref{thm:lower-bound-OBDD1} and Theorem \ref{thm:lower-bound-OBDD2}, the same bound is also valid for exact POBDDs and UOBDDs.

\begin{theorem}
	An exact id-AfOBDD $M$ with $ n $ classical and $ n $ affine states can solve $ \hwb $.
\end{theorem}
\begin{proof}
	The classical states are $ s_0,\ldots,s_{n-1} $ where $s_0$ is the initial and only accepting state. The affine states are $ e_0,\ldots,e_n $ where $e_0$ is the only accepting affine state and $ v_0 $ is $ e_0 $.
    
    Until reading the last input bit ($x_n$), for each value 1, the index of the classical state is increased by 1. Meanwhile, the value of $ x_i $ is written to the value of $ e_i $ in the affine state: the affine state after the $(i-1)$-th step becomes $ v_{i-1} = (\overline{1}~~x_1~~\cdots~~x_{i-1}~~0~~\cdots~~0)^T $, where $ \overline{1} = 1 - \sum_{j=1}^{i-1} x_j $. If $ x_i=0 $, then the identity operator is applied. If $ x_i=1 $, then the following affine transformation is applied
    \[
    	\left( 
        \begin{array}{c|c}
        	\begin{array}{c|c|c}
    		0 & -1 ~~ \cdots~~ -1 & 0 
            \\ \hline
            0 & & 0
            \\
            \vdots & \mathbf{I_{(i-1) \times (i-1)}} & \vdots
            \\
            0 & & 0
            \\ \hline 
            1 & 1 ~~~~ \cdots ~~~~ 1 & 1
    	\end{array}
        & ~~~~ \mathbf{0} ~~~~
        \\ \hline
        \\
        \mathbf{0} & \mathbf{I}
        \\ \\
        \end{array} 
         \right),
    \]   
    which updates the values of $e_0$ and $e_i$ as $ - \sum_{j=1}^{i-1} x_j $ and 1, respectively, and does not change the other entries. Remark that the first entry can also be written as $ 1 - \sum_{j=1}^{i} x_j $. 
    
    Thus, before reading $x_n$, the classical state is $ s_t $ for $ t =\sum x_1 + \cdots+ x_{n-1} $, and, the affine state is
    $
    	v_{n-1} = (\overline{1}~~x_1~~x_2~~\cdots~~x_{n-1})^T,
   $
    where $ \overline{1} $ is $ 1 - \sum_{j=1}^{n-1} x_j $. 
    
    After reading $x_n$, the classical state is always set to $s_0$ and so the decision is made based on the affine state. The pair $s_t$ and $x_n$ determines the last affine transformation that sets the final affine state to
    $
    	v_f = (x_z~~1-x_z~~0~~\cdots~~0)^T.
   $
    Here (i) if $ t=0 $ and $ x_n = 0 $, $x_z = 0 $, (ii) if $ t=n-1 $, then regardless the value of $ x_n $, $ x_z =1 $, and (iii) in any other case $ z = t + x_n $ and $ x_z $ is set to the $ (z+1) $-th entry of $ v_{n-1} $. Thus, if $ x_z=0 $ ($x_z=1$), then the input is accepted with probability 0 (1).
       \qed
\end{proof}


\newcommand{\ws}{\mathtt{WS_n}}
\newcommand{\mws}{\mathtt{MWS_n}}

For a positive integer $n$, let $x,y$ be inputs of size $n$, let $p(n)$ be the smallest prime number greater than $n$, and let $s_n(x) = \left( \sum_{i=1}^{n} i \cdot x_i \right) \Mod  p(n) $. 

We use two functions \cite{s06}.
The weighted sum function $\ws(x) :\{0,1\}^n \rightarrow \{0,1\} $ is defined as $ \ws(x)  = x_{s_n(x)}$ if $s_n(x) \in \{1, \dots , n\}$, and $0$ otherwise. The mixed weighted sum function $\mws : \{0,1\}^n \times \{0,1\}^n \rightarrow \{0,1\} $ is defined as $  \mws(x, y) = x_i \oplus y_i$ if $i = s_n(x) = s_n(y) \in \{1, \dots , n\}$, and $0$ otherwise.

Any bounded-error POBDD or UOBDD solving $ \mws $ has a width of at least $ 2^{\Omega(n)}/n $ \cite{sz2000,s06}. The same bound is also valid for OBDDs, and so for exact POBDDs and UOBDDs as well. 

\begin{theorem}
	An exact id-AfOBDD can solve $ \ws $ with $ p(n) $ classical states and $ n $ affine states.
\end{theorem}
\begin{proof}
	We use almost the same algorithm given in the previous proof. The affine part is the same. The new classical states are $s_0,\ldots,s_{p(n)-1} $ with the same initial and single accepting state $ s_0 $.
    
    Until reading $x_n$, the same affine transitions are applied:
    
    $
    	v_{n-1} = (\overline{1}~~x_1~~x_2~~\cdots~~x_{n-1})^T,
    $
    where $ \overline{1} $ is $ 1 - \sum_{j=1}^{n-1} x_j $. 
    The classical transitions are modified as follows: The classical state before reading $ x_i $ ($i<n$), say $s_j$, is set to $ s_{j+i \cdot x_i \Mod p(n)} $. Thus, before reading $x_n$, the classical state is $ s_t $ where $ t = \left( \sum_{i=1}^{n-1} i \cdot x_i \right) \Mod  p(n)  $.
    
    At this point, the pair $ s_t $ and $ x_n $ sets the affine state to
    
    $
    	v_f = (x_{s_n(x)}~~1-x_{s_n(x)}~~0~~\cdots~~0)^T
    $ and then the classical state is set to $ s_0 $. Here $ x_{s_n(x)} = 0 $ if $ t+n \cdot x_n \Mod p(n) \notin \{1,\ldots,n\} $. Otherwise, depending on the value of $ s_n(x)$, $ x_{s_n(x)}$ is set to the corresponding value from $ v_{n-1} $ or directly to $ x_n $. 
\qed
\end{proof}

\begin{theorem}
	An exact id-AfOBDD can solve $ \mws $ with $ p^2(n) $ classical states and $ (n+1) $ affine states.
\end{theorem}
\begin{proof}
	The classical states are $ \{ s_{i,j} \mid 0 \leq i,j < p(n) \} $ where $ s_{0,0} $ is the initial state and the accepting states are $ \{ s_{i,i} \mid 1 \leq i \leq n \} $. By reading all bits, the values of $ s_n(x) $ and $ s_n(y) $ are calculated and stored in the indexes of classical states in a straightforward way. Let $ s_{i,j} $ be the final state. It is clear that if $ i \neq j $ or $ i=j $ but not in $ \{1,\ldots,n\} $, the input is rejected classically. In the remaining part, we assume that the final classical state is $ s_{i,i} $ with $ i \in \{1,\ldots,n\} $. Thus, the final decision is given based on the final affine state.
    
    The affine states are $ \{e_0,\ldots,e_n\} $ where $ e_0 $ is the single accepting state. The initial affine state is $ v_0 = e_0 $.
       
    While reading the first part of the input, all $ x $ values are encoded in the affine state as
   $
    	(\overline{1}~~(-1)^{x_1}~~(-1)^{x_2}~~\cdots~~(-1)^{x_n})^T ,
    $
    where $ \overline{1} = 1 - \sum_{i=1}^n (-1)^{x_i} $.
    Then, for each $ y_j $ ($ 1\leq j \leq n $), we multiply the $ j $-th entry of the affine state by $ (-1)^{y_j} $. Thus, the $ j $-th entry becomes $ (-1)^{x_j+y_j} $, which is equal to $1$ if $ x_j=y_j $ and $-1$ if $ x_j \neq y_j $.
    
    The last affine transformation is composition of three affine transformations. The first transformation is the one explained above.  By using the second one, the affine state is set to
    $
    	( (-1)^{x_i+y_i} ~~~~ 1-(-1)^{x_i+y_i} ~~~~0~~\cdots~~0)^T.
   $
   
   Here the first two entries are $ (1~~0) $ for members and $ (-1~~2) $ for non-members. By using the third transformation, we can add the half the of second entry to the first entry and so get respectively $ (1~~0) $ and $ (0~~1) $. Thus, the AfOBDD can separate members from non-members with zero error.
   \qed
\end{proof}

\section{Las Vegas POBDDs and UOBDDs}\label{sec:lv-obdd}
For OBDDs with $ \varepsilon = \frac{1}{2} $ the lower bound given in Section \ref{sec:lower-bounds} can be at most quadratic. Up to a logarithmic factor, this quadratic gap was achieved by using the $\mathtt{SA_d}$ function in \cite{ss2005} for id-OBDDs. Here we give the same result for OBDDs (for any order) using the $\SSA$ function and also provide an LV-UOBDD algorithm with the same size as the LV-OBDD.

We start with the well-known {\em Storage Access} Boolean Function $\mathtt{SA_d}(x,y) = x_y$, where the input is split into the {\em storage} $x=(x_1,\dots,x_{2^d})$  and the {\em address} $y=(y_1,\dots,y_d)$.

By using the idea of ``Shuffling'' from \cite{a97,ak96,agky14,agky16} we define the {\em Shuffled Storage Access} Boolean function $\SSA:\{0,1\}^n \to \{0,1\}$, for even $n$. Let $x\in\{0,1\}^n$ be an input. We form two disjoint sorted lists of even indexes of bits $I_0=(2i_1,\dots,2i_m)$ and $I_1=(2j_1,\dots,2j_k)$ with the following properties:
(i)
if $x_{2i-1}=0$ then $2i\in I_0(x$,
(ii)
if $x_{2i-1}=1$ then $2i\in I_1(x)$,
(iii)
$i_r<i_{r+1}$, for $r\in\{1,\dots,m-1\}$ and $j_r<j_{r+1}$ , for $r\in\{1,\dots,k-1\}$.

Let $d$ be such that $2^d+d=n/2$.
We can construct two binary strings by the following procedure: initially $\alpha(x):=0^{2^d}$, then for $r$ from $1$ to $m$ we do $\alpha(x):=\mathtt{ShiftX}(\alpha(x), x_{2i_r})$, where $\mathtt{ShiftX}((a_1,...,a_m),b)=(a_2,...,a_m,a_1 \oplus b)$. 
And initially $\beta(x):=0^{d}$, then for $r$ from $1$ to $k$: $\beta(x):=\mathtt{ShiftX}(\beta(x), x_{2j_r})$. Then, $\SSA(x)=\mathtt{SA_{d}}(\alpha(x),\beta(x))$.

Firstly, we provide a lower bound for OBDDs. (See Appendix \ref{app:OBDD-SSA} for the proof.)
\begin{theorem}
	\label{thm:OBDD-SSA}
	$\OBDD{\mathtt{SSA_n}}\geq 2^{2^d}$, for $2^d+d=n/2$.
\end{theorem} 

Now, we provide an upper bound for LV-OBDDs.
\begin{theorem}
$\LVOBDDhalf{\mathtt{SSA_n}}\leq 2^{2^d/2+d+3}$, for $2^d+d=n/2$.
\end{theorem} 
\begin{proof}
The proof idea is to store and process even and odd bits of \textit{storage} separately in two different paths and then the path containing the index makes the correct decision and the other path ends with the answer of ``don't know''. The formal proof is as follows.

Firstly, we describe the states. Associate each state of each {\em odd} level with a quartet $(s,t,e,a)$, where $s\in\{0,\dots, 2^{2^d/2}-1\}$ 
is the string of stored bits, $t\in\{0,1\}$ is a Boolean variable marking whether we should store the next bit or not, $e\in\{0,1\}$ is a Boolean variable marking whether we store odd bits or bit or not  and $a\in\{0,\dots,2^d-1\}$ is the address. 
 We associate each state of each {\em even} level with a quintet $(s,t,e,a,b)$, where $s\in\{0,\dots, 2^{2^d/2}-1\}$, $t,e\in\{0,1\}$ and $a\in\{0,\dots,2^d-1\}$ have the same meaning, and $b\in\{0,1\}$ is the odd input bit in the current pair to know whether we should process the next bit as a {\em storage} bit or an {\em address} bit. We have the same set of states in both odd and even levels, but we do not use some states in odd levels.

On the first level the program reads the bit $x_{1}$ and starts in one of the $(0,0,0,0,x_{1})$ and $(0,1,1,0,x_{1})$ with equal probability $\frac{1}{2}$, and so chooses whether the algorithm will store odd or even bits of storage.
Let us describe the transitions of other levels. First, let us consider an {\em odd} level for a variable $x_i$. The program goes from the state $(s,t,a)$ to the state $(s,t,e,a,x_i)$. Second, let us consider an {\em even} level for a variable $x_i$. Being in a state $(s,t,e,a,0)$ means that the current bit is a {\em storage} bit. If $t=0$, then the program should skip the bit and go to the state $(s,1-t,e,a)$ on the next level. If $t=1$, then the program should store the bit by going to the state $(\mathtt{ShiftX}_{2^d/2}(s,x_i),1-t,e,a)$ on  the next level, where $\mathtt{ShiftX}_{u}(v,b)$ converts the integer $v$ to a binary string $v$ of length $u$ and applies the $\mathtt{ShiftX}(u,b)$ operation, then converts the result back to an integer. Being in a state $(s,t,e,a,1)$ means that the current bit is an {\em address} bit. The program then goes to the state $(s,t,e,\mathtt{ShiftX}_{d}(a,x_i))$.
On the last level, if $e = a \Mod 2$ then the {\em address} contains the index of a stored bit and the program returns its value, otherwise the program returns ``don't know''. 

The width of the program: $\mathtt{width}(P)=2^{2^d/2}\cdot 2\cdot 2\cdot 2^d\cdot 2= 2^{2^d/2+d+3}$.
\qed
\end{proof}

\begin{theorem}\label{thm:ulvobdd-ssa}
$\ULVOBDDhalf{\mathtt{SSA_n}}\leq 2^{2^{d}/2+d+3}$, for $2^d+d=n/2$.
\end{theorem} 
\begin{proof}(Sketch. Full proof is in Appendix \ref{app:ulvobdd-ssa})
We can construct a Las Vegas UOBDD from the probabilistic version, using Hadamard transformation instead of choosing states with equal probability. We use he same trick in next section in Theorem \ref{thm:lv-ufa}. Note that $\mathtt{ShiftX}_u(v,b)$ is a reversible operation and can be implemented as a unitary operator. 
\qed
\end{proof}

We note that affine OBDDs can also be exponentially narrower for $\SSA$. 
\begin{theorem}
	\label{thm:affine-SSA}
	An exact AfOBDD $A$ can solve $ \SSA $ with $ 2^{d+1} $ classical states and $ 2^{d} +1 $ affine states, for $2^d+d=n/2$.
\end{theorem}
\begin{proof}
The set of classical states is $\{ (p,s) \mid p \in \{p_{0},p_1\} \mbox{ and } s \in \{ s_{0},\ldots,s_{2^d-1} \} \}$, for a total of $2^{d+1}$ states. The states $p_j$ denote whether the next even bit to read is a storage or address bit. The state $s_{i}$ corresponds to the current address being $i$. Every classical state is accepting, and so the decision is given based on the final affine state. 
    
    The AfOBDD also has $ 2^d+1 $ affine states, $ \{e_1,\ldots,e_{2^d+1}\} $. The initial state is $ e_{2^d+1} $ ($v_0 = \left(0~~\cdots~~0~~1 \right)^T$) and the only accepting state is $ e_1 $.

	During the computation, it keeps the value of the storage in the first $ 2^d $ states. Specifically, if the next position to read is odd, or even but an address bit (we are in a state $(p_1,s_i$)) we perform the identity transformation on the affine state, changing only the classical state.
    If we are reading a storage bit, the classical state remains unchanged and we implement the $\mathtt{ShiftX}$ operation on the storage: first, the  $ 2^d $ entries are shifted to the left by one and the first entry becomes the $2^d$-th entry. Then, depending on the scanned symbol, the value of the $2^d$-th entry is updated:
	\begin{itemize}
	 \item If the scanned symbol is $ 0 $, then, for calculating the $ XOR $ value, the $2^d$-th entry is multiplied by 1, i.e., $ 0 \rightarrow 0 $ and $ 1 \rightarrow 1 $.
	 \item If the scanned symbol is $ 1 $, then, for calculating the $ XOR $ value, the $2^d$-th entry is multiplied by $-1$ and then  1 is added to this result, i.e., $ 0 \rightarrow 0 \rightarrow 1 $ and $ 1 \rightarrow -1 \rightarrow 0 $.
	 \end{itemize} 	
	 
     The last entry in the affine state is used to make the state vector well-formed. For example, if the state vector has $ 0 \leq t \leq 2^d $ 1s in its first $ 2^d $ entries, then the last entry is $ 1-t $.
	 
	 After reading the whole input, the first $2^d$ entries keep the storage. We know the address $i$ from our classical state $s_{i}$ -- we move the corresponding storage value from $e_i$ to $e_1$ and sum all other entries in $e_2$. Then, if the first entry is 1, the rest of the vector contains only 0. If the first entry is 0, the second entry is 1 and the rest of the vector contains 0.
	 
	 Therefore, any member of $\SSA$ is accepted by $ A $ with probability 1 and any non-member is accepted by $ A $ with probability 0.\qed
\end{proof}
\section{Las Vegas Automata and Zero-Error AfAs}
\label{sec:lv-FA}


\newcommand{\ENDk}{\mathtt{END_k}}

Similar to OBDDs, for $ \varepsilon = \frac{1}{2} $, the lower bound for finite automata given in Section \ref{sec:lower-bounds} can be at most quadratic. Up to a constant, this quadratic gap is achieved by using the language $ \ENDk = \{ u1v \mid u,v \in\{0,1\}^* \mbox{ and } |v| = k-1 \} $:

$
	\DFA{\ENDk} = 2^k \mbox{ and } \LVhalf{\ENDk} \leq 4\cdot 2^{k / 2}.
$

Here, we propose a new language $ \ModXOR $ based on which we improve the above constant for LV-PFAs and provide a LV-UFA algorithm with the same size as the LV-PFA.
Then, we show that an AfA can recognize it with exponentially fewer states with zero error. The language $ \ModXOR $ for $ k>0 $ is formed by the strings
$
	\{0,1\}^{<2k} x_1\{0,1\}^{2k-1} x_2\{0,1\}^{2k-1} \cdots x_m \{0,1\}^{2k-1}
$
where $ m > 0 $, each $ x_i \in\{0,1\} $ for $ 1 \leq i \leq m $, and $ \bigoplus_{i=0}^m x_i =1  $, taking $ x_0 = 0 $.

First, we give a lower bound for DFAs. (See Appendix \ref{app:dfa-modxor} for the proof.)

\begin{theorem}
	\label{thm:dfa-modxor}
	$ \DFA{\ModXOR} \geq 2^{2k} $ for each $ k > 0 $.
\end{theorem}

\begin{theorem}\label{th:lvpfa}
$ \LVhalf{\ModXOR} \leq 2\cdot 2^{k}$ for any $ k>0 $.
\end{theorem}
\begin{proof}
	We construct an LV-PFA, say $P$, with success probability $ \frac{1}{2} $. At the beginning of the computation, it splits into two paths with equal probability and in each path only deterministic transitions are implemented. In the first (second) path, it assumes that $ x_i $'s are placed on odd (even) indexes. The set of states is formed by
	$
		\{ (p,s) \mid p \in \{p_{0},p_1\} \mbox{ and } s \in \{ s_{(0)},\ldots,s_{(2^k-1)} \} \}
	$
	where the index of $ s_{(i)} $, $ (i) $, is a binary string of length $ k $ that represents the value of $ i $. The number of states is $ 2 \cdot 2^k $.
	
	The first path starts in $ (p_1,s_{(0)} )$ and the second path starts in $ (p_0,s_{(0)}) $. Both paths use the same transitions. If $ P $ is in a state $ (p_0,s) $, then it always switches to $ (p_1,s) $, i.e. the second part of the state does not change. If it is in $ (p_1,s_{(i)}) $ and reads a symbol $ x \in \{0,1\} $, then it switches to $ (p_0,s_{(j)}) $, where $ (j) $ is obtained from $ (i) $ as follows:
	\begin{itemize}
	\item $ (j)_{k} = x \oplus (i)_1 $ and
	\item $ (j)_l = (i)_{l+1} $ for $ 1 \leq l < k $.
	\end{itemize}
	Remark that all transitions in each path are reversible. It is clear that the first (second) path considers only the symbols at odd (even) indexes. 
    
	If the automaton ends in a state $ (p_0,s) $, then it says ``don't know'', since the initial guess was not correct and so the other path has the correct answer. If it ends in $ (p_1,s_{(i)}) $, then the decision is given based on the first digit of $ (i) $: the input is accepted (rejected) if $ (i)_1 = 1 $ ($ (i)_1 = 0 $).

	Lastly, since the initial state has $ s_{(0)} $, $ x_0 $ is correctly assumed to be $0$. Besides, if the length of the input is less than $ 2k $, then $ P $ never accepts since the first bit of $ (0) $ can not be changed before $ 2k $ steps.
\qed
\end{proof}

\begin{theorem}
	\label{thm:lv-ufa}
	$ \ULVhalf{\ModXOR} \leq 2 \cdot 2^k $ for any $ k > 0 $.
\end{theorem}
\begin{proof}
	We can follow the construction given for PFA. As pointed out in the above proof, both paths implement a deterministically reversible transitions and so they can be implemented by a UFA. Instead of an initial probabilistic distribution, the initial quantum state is set to
	$
		\frac{1}{\sqrt{2}} \ket{p_0,s_0} + \frac{1}{\sqrt{2}} \ket{p_1,s_0}. 
	$ 
	
	Then, both paths stay in superposition but each of them implements its part separately from the other. To be more precise, the unitary operator for a symbol can be defined as follows:
	\[
		\begin{array}{cc}
			& p_0 \times S ~~~~ p_1 \times S
			 \\
			\begin{array}{c}
				p_0 \times S \\ p_1 \times S
			\end{array}	
			&
			\begin{array}{c}
				\left( \begin{array}{c|c}
						~~~~~0~~~~~ & ~~~~~R~~~~~ \\ \hline I & 0 				
						\end{array}
				 \right)
			\end{array}							
		\end{array} ,
	\]
	where $ S = \{ s_{(0)},\ldots,s_{(2^k-1)} \} $. Since $ R $ is reversible (unitary), the overall matrix is unitary. The accepting states are $ \{ (p_1,s_{(1i')}) \} $ where $ i' $ is any binary string of length $ k-1 $.
	\qed
\end{proof}


Similarly to OBDDs, exact AfAs can also be exponentially more efficient than their classical and quantum counterparts.

\begin{theorem}
	The language $\ModXOR$ for $ k > 0 $ can be recognized by a ($2k+1$)-state AfA $ A $ with zero-error. 
\end{theorem}
\begin{proof}
	The AfA $ A $ does not use any classical state and it has $ 2k+1 $ affine states, $ \{e_1,\ldots,e_{2k+1}\} $. The initial state is $ e_{2k+1} $ and the only accepting state is $ e_1 $. It starts its computation in
	$
		v_0 = \left(0~~\cdots~~0~~1 \right)^T
	$.

	During the computation, it keeps the results in the values of the first $ 2k $ states, i.e. it sets the values 0 or 1 depending the previous results and the current scanning symbols. More specifically, before each transition, the first $ 2k $ entries are shifted to the right by one and the $2k$-th entry becomes the first entry. Then, depending on the scanned symbol, the value of the first entry is updated:
	\begin{itemize}
	 \item If the scanned symbol is $ 0 $, then, for calculating $ XOR $ value, the first entry is multiplied by 1, i.e., $ 0 \rightarrow 0 $ and $ 1 \rightarrow 1 $.
	 \item If the scanned symbol is $ 1 $, then, for calculating $ XOR $ value, the first entry is multiplied by $-1$ and then  1 is added to this result, i.e., $ 0 \rightarrow 0 \rightarrow 1 $ and $ 1 \rightarrow -1 \rightarrow 0 $.
	 \end{itemize} 	
	 The last entry in the affine state is used to make the state vector well-formed. For example, if the state vector has $ 0 \leq t \leq 2k $ 1s in its first $ 2k $ entries, then the last entry is $ 1-t $.
	 
	 After reading the whole input, the first entry keeps the result. The rest of entries are summed to the second entry. Thus, if the first entry is 1, then the rest of the vector contains only 0. If the first entry is 0, then the second entry is 1 and the rest of the vector contains zeros.
	 
	 Therefore, any member is accepted by $ A $ with probability 1 and any non-member is accepted by $ A $ with probability 0.   
	\qed
\end{proof}

{\bf Acknowledgements.} 
The work is partially supported by ERC Advanced Grant MQC and the Latvian State Research Programme NeXIT project No. 1.
The work is performed according to the Russian Government Program of Competitive Growth of Kazan Federal University.

\bibliographystyle{splncs03}
\bibliography{tcs}

\newpage

\appendix

\section{Basics of Branching Programs}
\label{app:basics-branching}

A branching program $P_n$ on the variable set $X=\{X_1,\ldots,X_n\}$ is a finite directed acyclic graph with one source node $ s $ and some sink nodes partitioned into two sets, Accept and Reject. Each inner node of $P$ is associated with a variable $X_i\in X$ and has two outgoing edges labeled $X_i=0$ and $X_i=1$. The program $P_n$ computes the Boolean function $f(x)$ ($f:\{0,1\}^n \rightarrow \{0,1\}$) as follows. For any $x\in\{0,1\}^n$ the computation starts from $ s $ and follows a path depending on the bits of $x$ -- in each vertex the related bit is tested and the next node is selected according to its value. If it ends in a sink node belonging to Accept (Reject), the input $x $ is accepted (rejected). For an accepted (rejected) input, the function takes the value 1 (0).

A branching program is {\em leveled} if the nodes can be partitioned into levels $V_1, \ldots, V_{\ell}$ and a final level $V_{\ell+1}$ such that (i) $s$ belongs to $ V_1 $, (ii) the nodes in $V_{\ell+1}$ are the sink nodes, (iii) nodes in each level $V_j$ with $j \le \ell$ have outgoing edges only to nodes in the next level $V_{j+1}$. 

The {\em width} of a leveled branching program $\texttt{width}(P_n)$ is the maximum of number of nodes in a level of $P_n$, i.e. $\texttt{width}(P_n)=\max_{1\le j\le \ell}|V_j| $. The {\em size} of a branching program $P_n$ is the total number of nodes.

A leveled branching program is called {\em oblivious} if all inner
nodes of one level are associated with the same variable. A branching
program is called {\em read-once} if each variable is tested on each
path only once. 

\section{Proof of Theorem \ref{thm:OBDD-SSA}}
\label{app:OBDD-SSA}

Let us consider any input order $\pi=(i_1,\dots,i_n)$. Let $\pi'=(j_1,\dots,j_{n/2})$ be only the even indexes according to $\pi$. Let $\Sigma\subset\{0,1\}^n$ be the set of inputs such that for any $\sigma\in\Sigma$ we have: $I_0(\sigma)=\{j_1,\dots,j_{2^d}\}$ and $I_1(\sigma)=\{j_{2^d+1}, \dots,j_{2^d+d} \}$. Informally, it means that the first $2^d$ even bits are the storage and the last $d$ even bits are the address.

Let us consider a partition of the input positions $\theta=(X_A,X_B)=\theta(\pi,u)$ such that $X_{j_r}\in X_A$, for $r\in\{1,\dots,2^d\}$ and $X_{j_r}\in X_B$, for $r\in\{2^d+1,\dots,2^d+d \}$. Informally, it means that the first $2^d$ even bits read belong to $X_A$ and the last $d$ even bits read belong to $X_B$. We now consider the partition of an input $\sigma=(\xi, \gamma)$ according to $\theta$. It is easy to see that $\beta(\sigma)$ depends only on $\gamma$. Let us denote it $\beta'(\gamma)=\beta(\sigma)$.

Let $K=\{\xi:(\xi,\gamma)\in\Sigma$ for some $\gamma\in\{0,1\}^{|X_B|}\}$.
We will show that for any different $\xi,\xi'\in K$ and mappings $\rho:X_A\to\xi, \rho':X_A\to\xi'$ the corresponding subfunctions $\SSA|_{\rho}$ and $\SSA|_{\rho'}$ are different. 

Let $\xi,\xi'\in K$ and $\xi\neq\xi'$, it means there is an even $j_r$ such that $\xi_{j_r}\neq\xi'_{j_r}$. Let us choose $\gamma$ such that $\beta'(\gamma)=(j_r/2)$. Let $\sigma=(\xi,\gamma)$ and $\sigma'=(\xi',\gamma)$. Then $\SSA(\sigma)=\xi_{j_r}\neq\xi'_{j_r}=\SSA(\sigma')$, therefore $\SSA|_{\rho}(\gamma)\neq \SSA|_{\rho'}(\gamma)$, so these are different subfunctions.

Hence $N^{\theta}(\SSA)\geq |K|=2^{2^d}$. Then, by definition, $N^{\pi}(\SSA)\geq N^{\theta}(\SSA)\geq 2^{2^d}$. This holds for any order $\pi$, therefore $N(\SSA)\geq 2^{2^d}$. We know from \cite{Weg00} that $N(\SSA)=\OBDD{\SSA}$, completing the proof.

\section{Proof of Theorem \ref{thm:ulvobdd-ssa}}
\label{app:ulvobdd-ssa}
The proof idea is to store and process even and odd bits of the \textit{storage} separately in two different paths and then the path containing the index makes the correct decision and the other path ends with the answer ``don't know''. The formal proof is as follows.

Firstly, we describe the states. Associate each state of each {\em odd} level with a quartet $(s,t,e,a)$, where $s\in\{0,\dots, 2^{2^d/2}-1\}$ 
is the string of stored bits, $t\in\{0,1\}$ is a Boolean variable marking whether we should store the next bit or not, $e\in\{0,1\}$ is a Boolean variable marking whether we store odd bits or bit or not and $a\in\{0,\dots,2^d-1\}$ is the address. 
 We associate each state of each {\em even} level with a quintet $(s,t,e,a,b)$, where $s\in\{0,\dots, 2^{2^d/2}-1\}$, $e,t\in\{0,1\}$ and $a\in\{0,\dots,2^d-1\}$ have the same meaning, and $b\in\{0,1\}$ is the odd input bit in the current pair to know whether we should process the next bit as a {\em storage} bit or an {\em address} bit. We have the same set of states in both odd and even levels, but we do not use some states in odd levels.

On the first level the program reads the bit $x_{1}$ and starts in the states $(0,0,0,0,x_{1})$ and $(0,1,1,0,x_{1})$ with equal amplitude $1/\sqrt{2}$, thus choosing whether the algorithm will store odd or even bits of the storage.

Let us describe the transitions of other levels. First, let us consider an {\em odd} level for a variable $x_i$. The program goes from the state $(s,t,e,a)$ to the state $(s,t,e,a,x_i)$. Second, let us consider an {\em even} level for a variable $x_i$. Being in a state $(s,t,e,a,0)$ means that the current bit is a {\em storage} bit. If $t=0$, then the program should skip this bit and go to the state $(s,1-t,a)$ on the next level. If $t=1$, then the program should store the bit by going to the state $(\mathtt{ShiftX}_{2^d/2}(s,x_i),1-t,a)$ on  the next level, where $\mathtt{ShiftX}_{u}(v,b)$ converts the integer $v$ to a binary string $v$ of length $u$ and applies the $\mathtt{ShiftX}(u,b)$ operation, then converts the result back to an integer. Being in a state $(s,t,e,a,1)$ means that the current bit is an {\em address} bit. The program then goes to the state $(s,t,e,\mathtt{ShiftX}_{d}(a,x_i))$.

After the last level the program is in the state $\frac{1}{\sqrt{2}}p_0+\frac{1}{\sqrt{2}}p_1$, where $p_0$ is the state corresponding to $(s,t,0,a)$ and $p_1$ is the state corresponding to $(s',t,1,a)$, where $s$ are even storage bits and $s'$ are odd storage bits.
Then the program measures qubits and gets one of the two states $p_0$ or $p_1$ with probability $1/2$. If the program gets $p_0$ and $a$ is even  or $p_1$ and $a$ is odd, then it can return the right answer. If the program gets $p_0$ and $a$ is odd or $p_1$ and $a$ is even, then it returns ``don't know''.

The width of the program is $\mathtt{width}(P)=2^{2^d/2}\cdot 2\cdot 2\cdot 2^d\cdot 2= 2^{2^d/2+d+3}$.

\section{Proof of Theorem \ref{thm:dfa-modxor}}
\label{app:dfa-modxor}

Let us consider $2^{2k}$ classes of words $C_1,\dots C_{2^{2k}}$.
Let $(a_0,\dots,a_{2k-1})=(z)$ be the binary representation of $z$. A word $w$ is in $C_z$ iff for any $j\in\{0,\dots,2k-1\}$ the following holds:
\[\bigoplus_{i\in\{1,\dots,|w|\}, (|w|-i) \Mod{2k} = j} w_i =a_j . \]

Let us show that two words from different classes are not equivalent by Myhill-Nerode.

Let $w\in C_z, w'\in C_{z'}$, $z\neq z'$. It means that for $(a_0,\dots,a_{2k-1})=(z)$ and $(a'_0,\dots,a'_{2k-1})=(z')$ there exists a $j$ such that $a_j\neq a'_j$. Without loss of generality, we can assume that $a_j=1$, $a'_j=0$.

Let us consider a word $v$ such that $|v|=2k-1-j$. Then, by the definition of the function, $wv\in \ModXOR$ and $w'v\not\in \ModXOR$.

Due to the Myhill-Nerode Theorem, the number of states required by a minimal DFA is the number of classes. Hence, $\DFA{\ModXOR}\geq 2^{2k}$.

\end{document}